\documentclass[10pt,twocolumn,letterpaper]{article}

\usepackage{btas}
\usepackage{times}
\usepackage{epsfig}
\usepackage{graphicx}
\usepackage{amsmath}
\usepackage{amssymb}
\usepackage{amsthm}
\usepackage{bm}
\usepackage{makecell}
\usepackage{multirow}
\usepackage{multicol}
\usepackage{algorithm}
\usepackage{algpseudocode}
\usepackage{comment}
\usepackage{mathtools}
\DeclarePairedDelimiter{\ceil}{\lceil}{\rceil}
\graphicspath{./Images}
\newtheorem{prop}{Proposition}

\btasfinalcopy 

\ifbtasfinal\fi

\begin{document}

\title{Securing CNN Model and Biometric Template using Blockchain}

\author{Akhil Goel$^*$, Akshay Agarwal$^*$, Mayank Vatsa$^*$, Richa Singh$^*$, and Nalini Ratha$^+$\\
$^*$IIIT Delhi and $^+$IBM Research, NY, USA\\
{\tt\small $^*$\{akhil15126, akshaya, mayank, rsingh\}@iiitd.ac.in, $^+$ratha@us.ibm.com }
}

\maketitle

\begin{abstract}
Blockchain has emerged as a leading technology that ensures security in a distributed framework. Recently, it has been shown that blockchain can be used to convert traditional blocks of any deep learning models into secure systems. In this research, we model a trained biometric recognition system in an architecture which leverages the blockchain technology to provide fault tolerant access in a distributed environment. The advantage of the proposed approach is that tampering in one particular component alerts the whole system and helps in easy identification of `any' possible alteration. Experimentally, with different biometric modalities, we have shown that the proposed approach provides security to both deep learning model and the biometric template.

\end{abstract}


\section{Introduction}
Recent advances in biometrics have made face, fingerprint, and iris authentication systems ubiquitous. These modalities are popular choices for large scale person identification. One example of such a scheme is India's Aadhaar project\footnote{https://uidai.gov.in/} where officials collect and store the biometric data of citizens to identify individuals at later stages. In these nation wide projects, there are two possibilities for data storage, either in a central server or in a distributed system. Storing all the collected data on a central server is reasonable in the sense that it requires the involvement of minimum people and hence chances of an external attack such as template tampering are minimum. However, this modeling choice comes at risks including system failure or an unexpected disk error that may cause some of the data to get corrupted. The other modeling option is to have a distributed system. Distributed systems ensure minimum loss in case of system failure but are hard to trust because of the number of entities involved in the process (mutual distrust between entities is a common problem in a distributed framework). An ideal solution to this conundrum is to guarantee a sense of trust and consensus in the distributed setting. This problem was first addressed by Nakamoto \cite{bitcoin} in the context of the digital currency, Bitcoin. Nakamoto proposed to replace trust with cryptographic proofs that validate a particular transaction in a blockchain. Over the years, many protocols have been developed that address the problem of achieving consensus in an unreliable and distributed environment.
\begin{figure}[t]
\begin{center}
   \includegraphics[width=1.0\linewidth]{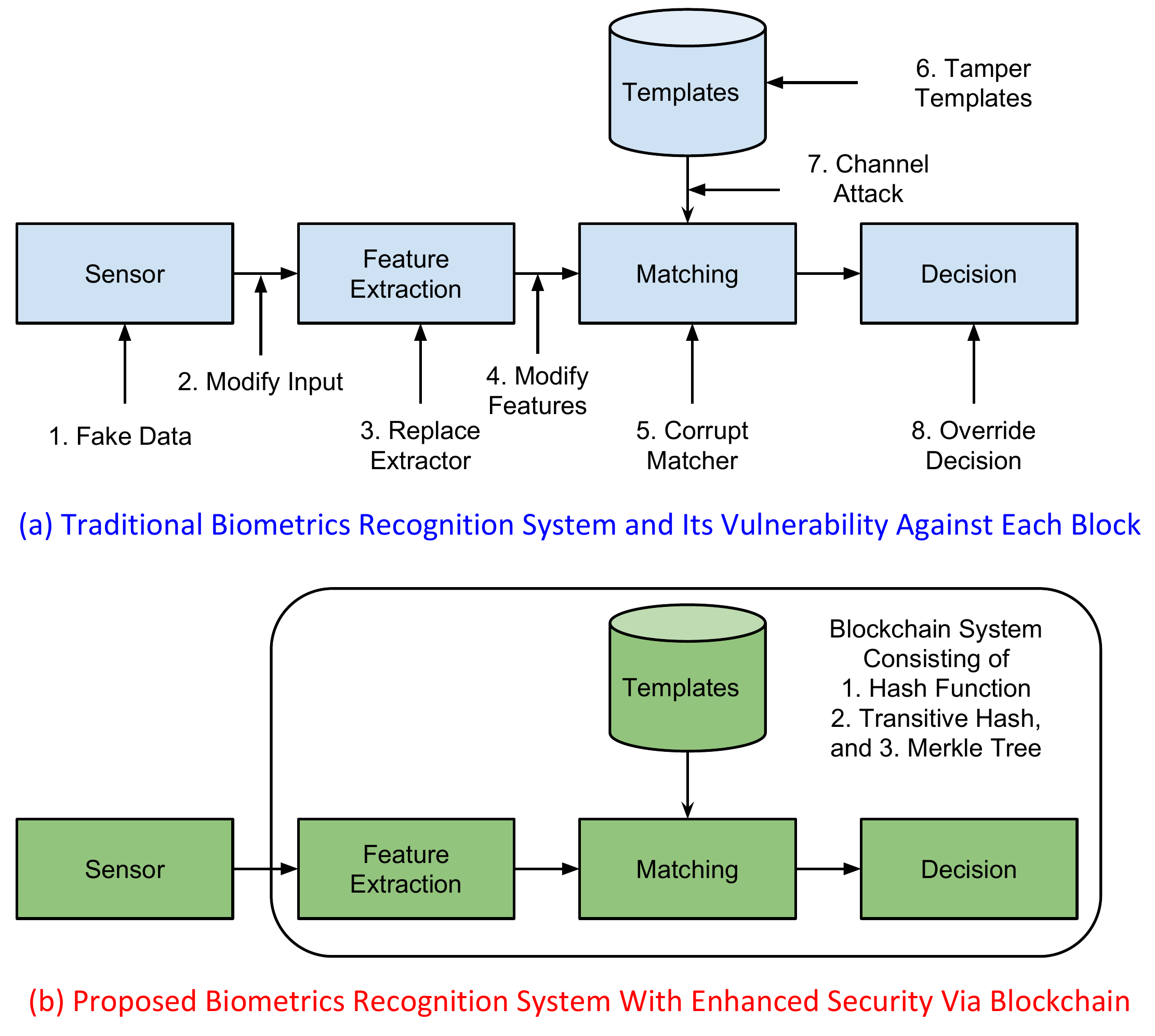}
\end{center}
\vspace{-3mm}
   \caption{Traditional biometrics recognition system which is vulnerable at multiple points and the proposed approach to secure the system.} 
\label{fig:motive}
\vspace{-5mm}
\end{figure}

As identified by Ratha et al. \cite{ratha2001enhancing}, biometric systems are vulnerable to external attacks. They identified various attack/vulnerability levels in a typical biometrics recognition pipeline including input, feature, matching, and decision stages. Figure \ref{fig:motive} shows the basic blocks of the biometrics recognition system and points where the attacks can be performed. A lot of research work has been done in crafting and defending against adversarial attacks at input level \cite{8294186,6990726}. For example, Carlini and Wagner \cite{carlini2017towards} and Chen et al. \cite{chen2018ead} have devised adversarial attacks that produce imperceptible yet deleterious input perturbations in a white-box framework. Researchers have proposed techniques to detect and mitigate the effects of these input level perturbations on facial images \cite{agarwal2018agnostic,ijcv2019goswami,goswami2018unravelling}. Goel et al. \cite{goel2018smartbox} have developed a toolbox to benchmark adversarial detection and mitigation algorithms for face recognition. 

Recently, researchers have been analyzing the implications of incorporating decentralized blockchain technology with biometrics systems. Butchman et al. \cite{buchmann2017enhancing} have proposed the security of breeder documents through blockchain. The biometric data captured from the subject is stored on the breeder document; the hash function of the report is computed and saved as a block in the blockchain. Nandakumar et al. \cite{nandakumar2017secure} have proposed the implementation of a concept called biometrics token. The token can be used a single time, and deduplication and multiple uses of the token can be avoided through blockchain hash function. The combination of biometrics and blockchain is in its novice phase. Gracia \cite{GARCIA20185} have shown the possibility of combining biometrics through blockchain to be implemented on a distributed system.
Similarly, Delgado-Mohatar et al. \cite{Oscar2019} have presented a view of how biometrics and blockchain can be combined to benefit each other. Zhou et al. \cite{zhou2018simple} have proposed the simple, auditable scheme for authentication of the fingerprint-based recognition system. These recent works have introduced the concept of blending these technologies for better security purposes. Other than the combination of biometrics and blockchain, in different fields such as smart energy \cite{aggarwal2018energychain, magnani2018feather}, healthcare \cite{gordon2018blockchain}, and intelligent devices \cite{stokkink2018deployment} blockchain have shown some advantages.

Progressing the efforts to enhance security and attaining consensus in a suspicious environment, in this research, we present a biometric recognition architecture which uses a private blockchain to extract features.  The extracted features are passed to a Merkle-tree type structure which performs template matching in a decentralized yet secure environment. The salient contribution of this research is the development of a secure deployable system which uses the parameters of a learned model and an amalgamation of biometrics and blockchain technology which is not only auditable but also transparent in its decision making task. 

\section{Blockchain}  Blockchain is a combination of blocks where each block consists of data and metadata such as the hash of the data and pointer to the hash of the previous block. The essential primitives of the blockchain include cryptography, hash function, transitive hash pointers, and digital signatures.  Cryptography provides the way for secure communication in the presence of an adversary which can steal the sensitive data stored in the device.  The hash function helps in mapping the data of arbitrary size to a fixed size of encrypted data. The significant properties of a hash function which makes them desirable for implementation are: (i) collision avoidance: two different inputs cannot produce the same output and (ii) deterministic hiding: the random encrypted data will match to the corresponding information even if they look entirely random. These are the properties which make the blockchain tamper-resistant. The transitive hash function helps in identifying the locations of possible changes in the data. The transitive nature can be seen in Figure \ref{fig:thash}. Due to the connection of blocks through the hash function, any modification in the data will lead to the change in the hash function of the data and ultimately the final hash. The hash pointers also ensure the integrity of the `ledger'.

\begin{figure}[t]
\begin{center}
   \includegraphics[width=1.0\linewidth]{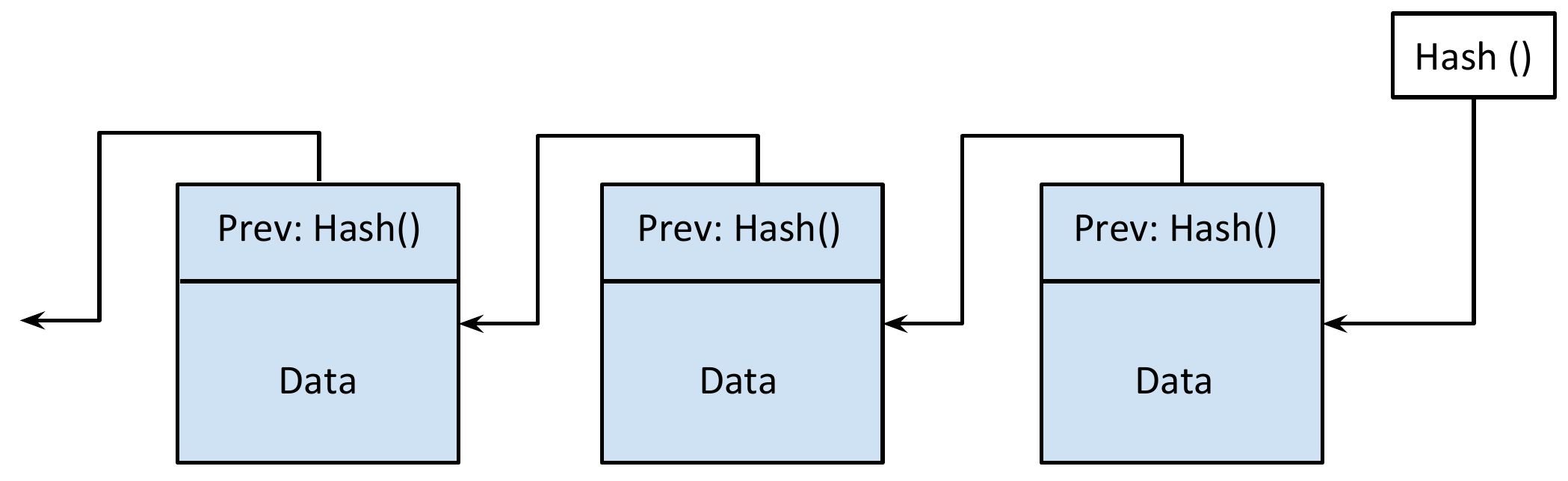}
\end{center}
\vspace{-3mm}
   \caption{The transitivity of the hash function involved in the blockchain.}
\label{fig:thash}
\vspace{-5mm}
\end{figure}



The final primitive of the blockchain is the digital signature. In the digital signature, the owner encrypts the data using the private key and can only be validated through the corresponding public key. The authenticated person is held responsible for signing the data using the private key. Hard to counterfeit property of digital signature makes it fault-tolerant.
\begin{figure*}[t]
\begin{center}
   \includegraphics[width=1.0\linewidth]{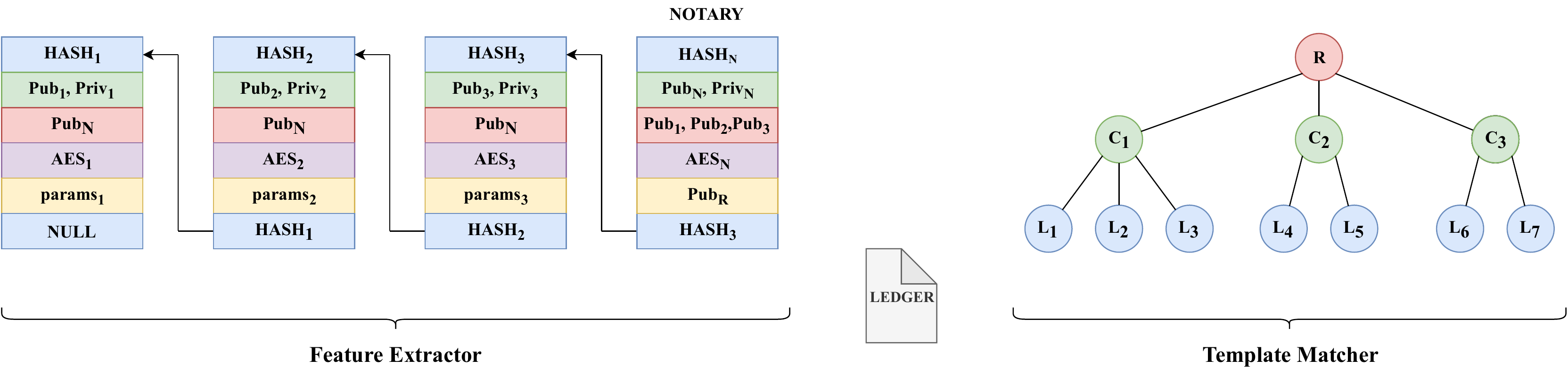}
\end{center}
\vspace{-3mm}
   \caption{Proposed biometrics recognition architecture with advanced features of blockchain.}
\label{fig:block}
\vspace{-3mm}
\end{figure*}


\section{Architecture of the Proposed Approach}
In this section, we propose an amalgamation of a private blockchain with the biometrics recognition system. 
The proposed architecture is shown in Figure \ref{fig:block} and has the following components:

\subsection{Feature Extractor}

Feature extractor extracts the matchable template from the supplied image. Each step of the feature extraction process is emulated as a single block of a private blockchain. The blocks in the chain may or may not follow the proper sequential order. The chain also has a notary block at the end which helps maintain sanity in the network. 

$i^{th}$ block of this chain contains the following:
\begin{itemize}
    \item $\mathbf{HASH_i : }$ Hash of the $i^{th}$ block
    \item $\mathbf{Pub_i, Priv_i : }$ Public and Private keys of the $i^{th}$ block
    \item $\mathbf{Pub_N : }$ Public key of the notary
    \item $\mathbf{AES_i : }$ AES key of the block
    \item $\mathbf{params_i : }$ Parameters required to complete the processing step
    \item $\mathbf{HASH_{i-1} : }$Hash of the previous block
    
\end{itemize}

Hash of any block is a function of the hash of the previous block and the parameters of the tasks performed by the current block. 

\begin{equation}
    HASH_i = \Phi(HASH_{i-1},params_i)
\end{equation}
Here $\Phi$ is any appropriate cryptographic hash function.

The feature extractor can be seen as a typical convolutional neural network (CNN) where each layer can be modeled as a block of the blockchain system \cite{akhilcvprw}.

\subsubsection{Notary Block}

Notary block is the last block of the feature extraction chain. It serves the following purposes: (i) it ensures that the information passes through the proper route and (ii) it makes sure that the data in the previous blocks has not been tampered. Notary block does not have any parameters of its own, and hence its hash is only a function of the hash of the previous block.

\begin{equation}
    HASH_N = \Phi(HASH_{N-1})
\end{equation}

Apart from its own AES, public and private keys, the block also holds a list of the public keys of all the blocks in the network in the proper sequential order and the public key of the root of the template matcher tree. 

\subsection{Template Matcher}


Template matcher matches the extracted features from the provided sample with the gallery of templates and returns the decision of the biometrics recognition process. The matcher is mirrored as a tree where just like a Merkle tree, the hash of the parent is the sum of the hash of all its children. The tree has the following levels:

\begin{itemize}
    \item Root level: This level contains only a single block called the root block. This block is responsible for declaring the final decision of the recognition process.
    
    \item Chief level: Blocks in this level consult with the leaves to reach a decision consensus and then send this decision to the root block.
    
    \item Leaf level: Blocks in this level contain the templates of the gallery and perform the actual template matching process. They share their results with their corresponding chief node to reach a final consensus decision.
\end{itemize}

Each root-chief link corresponds to an asymmetric decision key pair. The root divides the private key of the pair into $2*n+1$ shards using the Shamir's \cite{shamir1979share} secret sharing principle\footnote{Here $n$ is the maximum number of leaf nodes allowed.} A total of $n+2$ shards are required to reconstruct the key. The root keeps $n$ of these shards with itself and gives the remaining of the shards to the chief. The chief provides each leaf block with a single shard and keeps remaining of the shards with itself.


As explained above, hash of chief blocks would be the sum of hashes of all its leaves and hash of the root block would be the sum of hashes of all the chiefs. The root and the chief blocks also hold a copy of the original hash values of their children. This helps in tracking the faulty leaf in case of a tampering attack on the templates.

\section{Working of the Proposed Approach}

In this section we describe the working of both the feature extractor and the template matcher.

\subsection{Feature Extractor}

The captured data from the sensor is encrypted using the public key of the notary and updated on the ledger. Ledger after the update is represented as follows:

$$(encrypt_{Pub_N}(captured\_data))$$

The notary recognises an update of this sort as the start of a new query cycle. It decrypts the encrypted data using its own private key and encrypts it back using its AES key. 

\begin{equation}
    X_0=decrypt_{Priv_N}(encrypt_{Pub_N}(captured\_data))
\end{equation}

\begin{equation}
    ED_1=encrypt_{AES_N}(X_0)
\end{equation}

In order to allow only the next sequential block to be able to decrypt this value, the notary encrypts its own AES key using the public key of the next in line block.

Suppose the next step (first step in this case) in the extraction pipeline corresponds to block $i$, then:

\begin{equation}
    EK_1=encrypt_{Pub_i}(AES_N)
\end{equation}

It also encrypts a \textit{start} message using the public key of the next block to let it know that the update is meant for him. 
\begin{equation}
    EM_1=encrypt_{Pub_i}(start\_message)
\end{equation}

Lastly, the notary produces its own digital signature by signing a universal message using its own private key.

\begin{equation}
    S_1=encrypt_{Priv_N}(message)
\end{equation}

Notary updates the ledger with $(ED_1,EK_1,EM_1,S_1)$. After the ledger update, all the blocks in the chain verify if it is their turn to act. They do this by verifying the encrypted start message.

Each block $j$ in the chain asserts the following:

\begin{equation}
    assert(decrypt_{Priv_j}(EM_1)==start\_message)
\end{equation}

This assertion is true only for the $i^{th}$ (the intended) block. After the start message verification, the $i^{th}$ block verifies whether the update has been made by the notary. The block does not proceed further if this signature verification fails.

\begin{equation}
    assert(decrypt_{Pub_N}(S_1)==message)
\end{equation}

The above two assertions ensure the following:

\begin{itemize}
    \item Last ledger update was made by the notary.
    \item Notary has directed the $i^{th}$ block to continue with the computation.  
\end{itemize}

After successfully completing the processing steps, let the output produced be $O_1$. Ledger update by block $i$ corresponds to the following set of equations:

$$ED_2=encrypt_{AES_i}(O_1)$$
$$EK_2=encrypt_{Pub_N}(AES_i)$$
$$EM_2=encrypt_{Pub_N}(start\_message)$$
$$S_2=encrypt_{Priv_i}(message)$$

The notary picks up this update and proceeds in a similar approach. This cycle between the blocks and the notary continues until the final feature vector is produced. In this case, instead of using the public key of any of the blocks of the extractor chain, the notary updates the ledger using the public key of the template matching tree.

\subsection{Template Matching}

After successful feature vector computation, the root of the matching tree receives the encrypted version of the template. It decrypts it and forwards it to all of its chief blocks. The chief blocks further send it to all of their leaf blocks. Leaf blocks decrypt the feature vector and compare it with the stored template using an appropriate measure such as the Euclidean distance or the cosine similarity. Each leaf block shares their score with their chief block. The chief block analyses the scores of all the blocks and prepares a decision document with the identity and the rating of the template with the maximum match (e.g. minimum Euclidean distance score) to share with the root. Before sharing with the root, the chiefs share the document with the leaves for their consent. Leaves check if the distance score in the document is less than their score. If the score in the document is indeed less than their score, then from their perspective the document is proper. Leaves consent by providing the chief with their shard of the decision key. If however, the score in the decision document is greater than the score of the leaf block, then that particular block does not consent and turns its flag `ON'. After consenting with the leaves, the chief adds its shards to the pool and provides the final decision and the collected shard pool to the root. If the root reconstructs the private decision key by adding one more shard to the pool, then consensus has been reached. However, if that is not the case, the root checks the score of all the leaves with an `ON' flag and compares it with the score in the decision document to get the template with maximum similarity from that particular path. After receiving and validating the most similar templates from each of the chiefs, the root decides the match by choosing the most similar template. This distributed system provides transparency in the decision-making process.

\section{Security Analysis}
In this section, we discuss the security mechanism of the proposed architecture.

\begin{prop}
\label{prop:prop1}
A compromised chief node in the template matching tree can never receive consensus while drafting a decision document.

\end{prop}

\begin{proof}
A total of $n+2$ shards are required to achieve consensus in the decision-making process. $n+1$ of these shards come from the chief and the leaf blocks. The root provides the remaining $1$ shard. If the decision document is faulty, there would be at least one leaf block who would have received a distance score less than the one mentioned in the document. This block will not consent and will not give its shard to the chief block. Therefore, the maximum number of shards that the chief can collect is $n$. It would require the root to provide an additional shard to construct the decision key and in case of mismatch, consensus cannot be achieved.     
\end{proof}

\begin{itemize}
    \item \textbf{Adversaries try to query/intercept the feature extractor midway:} Several attack techniques have been devised which monitor the functioning of the feature extraction process by analyzing the magnitude and flow of the gradients to construct a sample which would result in model failure. For this, they either need to query or intercept the information flow midway. The proposed model does not allow either of the two. Querying the model midway would result in request rejection since the notary would not have signed the query. Intercepting the information flow midway would not work either since the attacker does not possess the required keys to decrypt the layer output.

    \item \textbf{Parameters of the feature extractor blocks are compromised:} Hash of a feature extractor block depends on both the parameters of the block and the hash of the previous block. If the block parameters get compromised, the hash of the block changes. It triggers a change in the hash of the next block and finally changes the hash of the notary. Hence, a change in the hash of the notary signifies that parameters of any of the block have been compromised. It can be rectified by figuring out the first block whose hash does not match with the hash that it had in the last stable state.
    
    \item \textbf{Templates are compromised:} Hash of the leaves of the template matching tree depend on the templates they hold. A change in them would correspond a change in the hash values which would correspond a change in the hash of the chief and ultimately the root. After detecting a change in the hash of itself, the root scrutinizes by finding the chief node whose hash change triggered its hash change. The faulty chief node repeats the same practice to find out the faulty leaf node. The faulty leaf node is forced to revert to its last stable state by restoring the correct template.
    
    \item \textbf{A chief block in the template matching tree is compromised:} A compromised chief block would try to pass over an incorrect decision document to the root. However, to get it passed, it would require the consensus of all the leaves. Proposition \ref{prop:prop1} shows that it is impossible for such a block to achieve consensus. It triggers scrutiny from the root who identifies the leaf with its flag turned on (most similar template).
    
    \item \textbf{A leaf block is compromised:} All the leaf blocks must consent to receive $n+2$ total shards. A compromised leaf block would not agree to a valid decision document and would trigger scrutiny from the root. The scrutiny would reveal that the decision document was valid and the decision-making process would go on.
\end{itemize}

\begin{table*}[t]
\centering
\caption{Rank 1 identification accuracy (\%) before and after template tampering on both the architectures.}
\label{table:table1}
\setkeys{Gin}{keepaspectratio}
\resizebox*{0.95\textwidth}{0.95\textheight} {
\begin{tabular}{|c|c|c|c|c|c|}
\hline
\multirow{3}{*}{\textbf{Dataset}} & \multirow{3}{*}{\textbf{Metric}} & \multicolumn{4}{c|}{\textbf{Rank 1 Accuracy (\%)}} \\ \cline{3-6} 
 &  & \multicolumn{2}{c|}{\textbf{Before Template Tampering}} & \multicolumn{2}{c|}{\textbf{After Template Tampering}} \\ \cline{3-6} 
 &  & \textbf{Traditional Architecture} & \textbf{Proposed Architecture} & \textbf{Traditional Architecture} & \textbf{Proposed Architecture} \\ \hline
\textbf{MultiPIE} & \textbf{Cosine} & 90.79 & 90.79 & 3.44 & 90.79 \\ \cline{2-6}  \hline
\textbf{MEDS} & \textbf{Euclidean} & 70.37 & 70.37 & 0 & 70.37 \\ \hline
\textbf{CASIA} & \textbf{Cosine} & 95.39 & 95.39 & 70.81 & 95.39 \\ \cline{2-6} \hline
\end{tabular}%
}\\
Results with best performing distance metric is reported.
\vspace{-5mm}
\end{table*}

\section{Experiments}


Experiments are performed to present the efficacy and security of the proposed model using face and fingerprint modalities. The results of face identification are reported using a subset of CMU Multi-PIE Face Database \cite{gross2010multi} and Multiple  Encounters Database (MEDS-II) database \cite{meds}. Multi-PIE database is a collection of 750,000 images of 337 subjects captured under 15 viewpoints and 19 illumination conditions in four different recording sessions. In this research, the subset of the database has been used with 50248 frontal images of all the subjects. The subset contains at least 40 frontal images for each identity. MEDS database is collected from 518 subjects and contains 1,309 frontal images of each subject. The results of fingerprint identification are reported on a subset of CASIA fingerprint database\footnote{http://biometrics.idealtest.org/index.jsp}. CASIA fingerprint database contains 500 fingerprints of each of the 4000 subjects. The subset considered for the experiments consists of 100 subjects with 10 fingerprints each. 

We use a pre-trained CNN network, VGG-Face model \cite{vggface}, as the feature extractor and model each of its layer as a block of the chain. The parameters of the block include the kernels, biases and the activation function for the convolution layers, pool size for the max-pooling layers and weights, biases and activation function for the dense layers. We use the python implementation of Shamir's secret sharing\footnote{https://github.com/lamby/python-gfshare} to implement key sharding.

 Each of the dataset is divided into training and testing sets in a 4:1 ratio using stratified random sampling. The VGG network is fine-tuned by freezing the weights of the first nineteen layers and re-learning the weights of the remaining layers of the model. We chose the fully connected dense layer with 4096 units as our template vector. For $n$ subjects, the constructed template matcher tree has $\ceil*{n//50}$\footnote{$//$ refers to integer division and $\%$ refers to remainder operation.} chief blocks. Each but the last block has 50 leaves and the last chief block has 50 leaves if $\ceil*{n//50}$ equaled $n//50$ and $n\%50$ leaves otherwise.
The gallery is constructed by choosing a random image of each subject and then by passing it through the trained model to generate the template. Each of the gallery templates is stored in a leaf of the template matcher tree. In our implementation, each leaf holds just one template corresponding to a particular subject in the dataset. Face and fingerprint identification is performed by matching the gallery and probe templates using the following distance metrics: (i) Cosine similarity and (ii) Euclidean distance.

To show the security and self-correcting nature of the system, the templates are perturbed using Gaussian noise. The effect of this perturbation is measured on both the traditional and the proposed architectures. Figure \ref{fig:graphsmeds} shows the cumulative match characteristics curves (CMC) before and after template tampering on both the architectures on the MEDS database. It is evident that template tampering reduces the performance of traditional designs significantly, whereas it does not affect the proposed model. Further, Table \ref{table:table1} shows the rank-1 identification accuracies on all three databases using traditional and proposed biometrics recognition system. On the Multi-PIE database, when the cosine distance measure is used for identifying, both the traditional and proposed architecture yield $90.79$\% accuracy. The advantage of the combination of blockchain in biometrics system can be seen when the tampering is performed. In the traditional network case, the model suffers more than $87$\% drop in the identification accuracy. On the other hand, due to the property of blockchain, the model and template are protected against tampering and same identification accuracy is maintained.

Similarly, on the MEDS face database and the CASIA fingerprint database, the traditional architectures yield $70.37$\% (euclidean) and 95.39\% (cosine) rank-1 identification accuracies, respectively. Due to the lack of any security features in the traditional architectures, the models suffer a huge performance drop after tampering. However, the proposed blockchain incorporated CNN models retain the original identification accuracies in both the case studies.


\section{Time Complexity}
In a distributed framework, the time needed by the template matcher will be: Time to delegate to the leaves + Time to conduct a template match + Time to compare $n$ values + Time taken to carry out Shamir’s secret ($x$) + Time taken to compare $c$ values ($n$: number of leaf nodes and $c$: number of chief nodes). Assuming constant time to delegate to the leaves, dimensions of the template as $a\times{b}$, and Euclidean distance as the matching score, the overall complexity is: O($ab$) + O($n$) + O($x$) + O($c$). The power consumption of the proposed system is similar to the traditional system.

\begin{figure}[]
\begin{center}
   \includegraphics[width=0.9\linewidth]{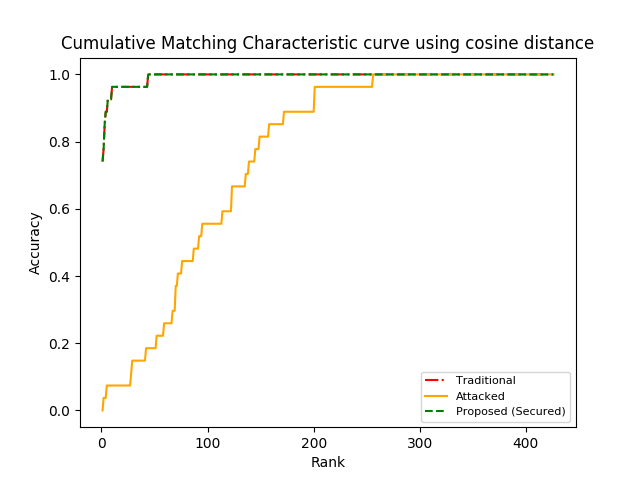}
\end{center}
\vspace{-3mm}
  \caption{Results of face identification(\%) on the MEDS dataset.}
\label{fig:graphsmeds}
\vspace{-3mm}
\end{figure}

\section{Conclusion}

In a ``security-trust model'', it is easier to trust a community, where a decision is made only when a majority of the community agrees, than to trust a particular individual. It is one of the motivations behind using a distributed framework for biometric template matching. 
Inspired from this, in this research, we have developed a self-correcting, template and parameter tamper-proof blockchain architecture for biometrics recognition. The proposed deep learning model is able to protect different stages of the biometrics recognition pipeline, specifically feature extraction, matching, and template storage. Experiments on face and fingerprint modalities showcase the effectiveness of the proposed approach. One of the key limitations of the proposed model is computation time because cryptographic computations, especially symmetric key encryption, and decryption are very demanding operations. They provide undeniable security but take some time to compute. As an extension to this work, we plan to develop a distributed framework which is not very resource hungry and has a time complexity similar to the traditional architecture.


\section{Acknowledgement}

A. Agarwal is partly supported by Visvesvaraya PhD Fellowship. M. Vatsa and R. Singh are partly supported from the Infosys Center for AI at IIIT-Delhi. M. Vatsa is also partially supported by the Department of Science and Technology, Government of India through the Swarnajayanti Fellowship.

{\small
\bibliographystyle{ieee}
\bibliography{0051}
}

\end{document}